\documentclass[11pt,letterpaper]{article}

\usepackage[margin=1in]{geometry}
\usepackage{algorithm}
\usepackage{algorithmic}
\usepackage{natbib}
\usepackage{hyperref,enumitem}
\hypersetup{
	colorlinks=true,
	citecolor = blue 
}

\usepackage{url}  
\usepackage{xspace}
\usepackage{graphicx}
\usepackage{amsmath,amssymb,amsthm}
\usepackage{booktabs} 
\usepackage{color}
\usepackage[font=small,labelfont=bf]{caption}
\usepackage{subcaption}
\usepackage{multirow}
\usepackage{authblk}

\usepackage{comment}
\usepackage{xcolor}
\usepackage{amsfonts} 
\usepackage{nicefrac} 
\usepackage{microtype} 

\usepackage[utf8]{inputenc} 
\usepackage[T1]{fontenc} 

\usepackage{tikz}
\usetikzlibrary{shapes}
\usetikzlibrary{arrows}
\usetikzlibrary{calc}
\usetikzlibrary{decorations.pathreplacing}

\newtheorem{theorem}{Theorem}

\theoremstyle{definition}
\newtheorem{definition}[theorem]{Definition}
\newtheorem{remark}[theorem]{Remark}

\newtheorem{example}[theorem]{Example}

\begin{document}

\title{Deep Learning for Two-Sided Matching}

\author[a]{Sai Srivatsa Ravindranath}
\author[a]{Zhe Feng}
\author[c]{Shira Li}
\author[b]{Jonathan Ma}
\author[c]{Scott D.~Kominers}
\author[a]{David C.~Parkes}

\affil[a]{John A.~Paulson School of Engineering and Applied Sciences, Harvard University \authorcr \texttt{saisr,zhe\_feng,parkes@g.harvard.edu}}
\affil[b]{Harvard College \authorcr\texttt{jonathan.q.ma@gmail.com}}
\affil[c]{Harvard Business School \authorcr \texttt{shirali@g.harvard.edu, kominers@fas.harvard.edu}}

\maketitle 

\begin{abstract}
We initiate the study of deep learning for the automated design of two-sided matching mechanisms. What is of most interest is to 
use machine learning to understand 
the possibility of new tradeoffs between {\em strategy-proofness} and {\em stability}. 
These properties cannot be achieved simultaneously, but the efficient frontier is not understood. 
We introduce novel differentiable surrogates for quantifying ordinal strategy-proofness and stability and use them to train differentiable matching mechanisms that map discrete preferences to valid randomized matchings. We demonstrate that the efficient frontier characterized by these learned mechanisms is substantially better than that achievable through a convex combination of baselines of {\em deferred acceptance} (stable and strategy-proof for only one side of the market), {\em top trading cycles} (strategy-proof for one side, but not stable), and {\em randomized serial dictatorship} (strategy-proof for both sides, but not stable). 
This gives a new target for economic theory and opens up new possibilities for machine learning pipelines in matching market design.
\end{abstract}

\section{Introduction}

Two-sided matching markets, classically used for settings such as high-school matching, medical residents matching, and law clerk matching, and more recently used in online platforms such as Uber, Lyft, Airbnb, and dating apps, play a significant role in today's world. As a result, there is a significant interest in designing better mechanisms for two-sided matching. 

The seminal work of~\cite{GS62} introduces a simple mechanism for stable, one-to-one matching in two-sided markets---\emph{deferred-acceptance} (DA)---which has been applied in many settings, including doctor-hospital matching~\citep{RothPeranson99}, school choice~\citep{abdulkadiroglu2003school,pathak2008leveling,APR2009}, and cadet-matching~\citep{sonmez2011matching,sonmez2011bidding}.
The DA mechanism is \emph{stable}, i.e., no pair of participants prefer each other to their match (or to being unmatched, if they are unmatched in the outcome). However, the DA mechanism is not \emph{strategy-proof} (SP),
and a participant can
sometimes misreport their preferences to obtain a better outcome
 (although it is SP
for participants on one side of the market).
Although widely used, this failure of SP for the DA mechanism 
presents a challenge for two main reasons. First, it can lead to unfairness, where better-informed participants
can gain an advantage in knowing which misreport
strategies can be helpful.
 Second, strategic behavior can lead to lower quality, unintended outcomes,
 and outcomes that are unstable with respect to true preferences.

In general, it is well-known
that there must necessarily be a tradeoff between stability and strategy-proofness: 
it is provably 
impossible for a mechanism to achieve 
both stability and strategy-proofness~\citep{DubinsFreeman1981,Roth82}. 
A second example of a matching mechanism 
is 
\emph{random serial dictatorship} (RSD)~\citep{AS98}, 
which is typically adopted for one-sided assignment problems rather than two-sided matching.
When adapted to two-sided matching, RSD 
is SP but not stable. In fact, 
a participant may even prefer to remain unmatched
than participate in the outcome of the matching.
A third example of a matching mechanism 
is 
the {\em top trading cycles} (TTC) mechanism~\citep{SHAPLEY197423},
also typically adopted for one-sided assignment problems rather than 
problems of two-sided matching.
In application to two-sided matching, TTC is neither SP nor stable (although it is SP
for participants on one side of the market).

There have been various research efforts to circumvent this impossibility result. Some relax the definition of strategyproofness~\citep{mennle2021partial} while others characterize the constraints under which stability and strategyproofness are achieved simultaneously~\citep{kamada2018stability, RePEc:oup:restud:v:88:y:2021:i:3:p:1457-1502.,10.1257/0002828054825466}. 
The tradeoff between these desiderata remains poorly understood beyond the existing point solutions of DA, RSD, and TTC. However, we argue that real-world scenarios demand a more nuanced approach that considers both properties. The case of the Boston school choice mechanism highlights the negative consequences of lacking strategy-proofness, resulting in unfair manipulations by specific parents~\citep{RePEc:nbr:nberwo:11965}.
At the same time, the importance of stability in matching markets is well understood~\citep{stable_unravel}.

Recognizing this, and inspired by the success of 
 deep learning in the
 study of revenue-optimal auction design~\citep{deep-auction19,curry2020certifying,STZ19,rahme2020permutationequivariant}, 
we initiate the study of deep learning for the design of two-sided matching mechanisms 
{\em We ask whether deep learning
frameworks can enable a systematic study of this tradeoff.}~By answering this question affirmatively, we open up the possibility of using machine learning pipelines to open up new opportunities for economic theory---seeking theoretical characterizations
of mechanisms that can strike a new balance between strategyproofness and stability. 

We use a neural network to represent the rules of a matching mechanism, mapping preference reports to a distribution over feasible matchings, and show how we can use an unsupervised learning pipeline to characterize the efficient frontier for the design tradeoff between stability and SP.
The main methodological challenge in applying neural networks to two-sided matching comes from handling the ordinal preference inputs (the corresponding inputs are cardinal in auction design) and identifying suitable, differentiable surrogates for approximate strategy-proofness and approximate stability.

We work with randomized matching mechanisms, for which the strongest SP concept is {\em ordinal strategy-proofness}. This aligns incentives with truthful reporting, whatever an agent's utility function (i.e., for any cardinal preferences consistent with
an agent's ordinal preferences). Ordinal SP is equivalent to the property of {\em first-order stochastic dominance} (FOSD)~\citep{Erdil14}, which 
suitably defines the property that 
an agent has a 
better chance of getting their top, top-two, top-three, and so forth choices when they report truthfully. 
As a surrogate for SP, we quantify during training the degree to which FOSD is violated. For this, we adopt an {\em adversarial learning approach}, augmenting the training data with defeating misreports that reveal the violation of FOSD.
We also define a suitable surrogate to quantify the degree to which stability is violated. This surrogate aligns with the notion of {\em ex ante} stability---the strongest stability concept for randomized matching. 

We train neural network representations of matching mechanisms 
by adopting stochastic gradient descent (SGD) for loss functions that 
are defined on different convex combinations of the two surrogate quantities,
and construct the efficient frontier for stability and strategy-proofness for different market settings.
A further challenge with ordinal preference inputs as opposed to cardinal inputs 
arises because ordinal preferences are discrete. 
In the present work, we simply
 enumerate the possible
misreports of an agent
as a step
when evaluating the derivative of error 
for a particular training example. In contrast, the use of 
deep-learning approaches
for the design 
of revenue-optimal auctions 
works with 
a continuous 
space of agent valuations
and 
gradient-ascent in this adversarial step of identifying
useful misreports~\citep{deep-auction19}.
As discussed at the end of the paper,
the challenge that this presents in
scaling to large numbers of agents 
can be resolved 
by assuming a suitable structure on preference orderings
in the domain,
or on the language that is
made available to agents in
reporting their preferences to the mechanisms.

Our main experimental results demonstrate that this novel
use of deep learning can strike a much better trade-off between stability and SP than that achieved by a convex combination of the DA, TTC, and RSD mechanisms. Taken as a whole, these results suggest that deep learning pipelines can be used to identify new opportunities for matching theory.
For example, we identify mechanisms
that are provably almost as stable as DA and yet considerably more strategy-proof. We also identify mechanisms
that are provably almost as strategy-proof as RSD and yet considerably more stable.
These discoveries raise opportunities for future work in economic theory, in regard to understanding
the structure of these two-sided matching mechanisms as well as 
characterizing the preference distributions for which this is possible.

\noindent \section{Related work.~~~} \citet{DubinsFreeman1981} and~\citet{Roth82} show the impossibility of achieving both stability and SP in two-sided matching. \cite{alcalde1994top} also show the impossibility of individually rational, Pareto efficient, and SP allocation rules, and this work has been extended to randomized matching~\citep{alva2020impossibility}. RSD is 
 SP but may not be stable or even individually rational (IR)~\citep{AS98}. We will see that the top trading cycles (TTC) mechanism~\citep{SHAPLEY197423}, when applied in a two-sided context, is only SP for one side, and is neither stable nor IR. 
 The DA mechanism~\citep{GS62} is stable but not SP; see also~\citep{Roth93}, who study the polytope of stable matchings. The \textit{stable improvement cycles} mechanism~\citep{erdil2008s} achieves as much efficiency as possible on top of stability but fails to be SP even for one side of the market. Finally, a series of results show that DA becomes SP for both sides of the market in large-market limit contexts~\citep{immorlica2015incentives,kojima2009incentives,lee2016incentive}.

\citet{10.2307/45212389} discuss different stability and no envy concepts. We focus on {\em ex ante} stability~\citep{KestenUnver15}, also discussed by~\citep{Roth93} as {\em strong stability}. \citet{mennle2021partial} discuss different notions of approximate strategy-proofness in the context of matching and allocation problems. In this work, we focus on ordinal SP and its analog of FOSD~\citep{Erdil14}. This is a strong and widely used SP concept in the presence of ordinal preferences.
There are a lot of other desiderata, such as efficiency, that are also incompatible with strategyproofness.~\citet{mennle2017hybrid} study this trade-off through hybrid mechanisms which are convex combinations of a mechanism with good incentive properties with another which is efficient. In the context of social choice, other work studies the trade-off between approximate SP and desiderata, such as plurality and veto voting~\citep{TimoSven16}.

\citet{ConitzerS02,ConitzerS04a} introduced the automated mechanism design (AMD) approach that framed problems as a linear program. 
However, this approach faces severe scalability issues as the formulation scales exponentially in the number of agents and items~\citep{GuoC10}. 
Overcoming this limitation, 
more recent work seeks to use deep neural networks to address problems of economic design~\citep{deep-auction19, deep-budget, deep-facility, curry2020certifying,STZ19,rahme2020permutationequivariant, pmlr-v162-duan22a, ivanov2022optimal}, but not until now to matching problems. 
As discussed in the introduction, 
two-sided matching brings about new challenges, most notably 
in regard to
working with discrete,
ordinal preferences and adopting the right surrogate loss functions for approximate SP and approximate stability. 
Other work has made use of support vector machines to search for stable mechanisms, but without considering strategy-proofness~\citep{NarasimhanParkes16a}. 
A different line of research
is also considering 
stable matching together with bandits problems,
 where agent preferences are unknown {\em a priori}~\citep{das2005two, pmlr-v108-liu20c, dai2021learning, 10.5555/3546258.3546469, pmlr-v139-basu21a, pmlr-v130-sankararaman21a, jagadeesan2021learning, pmlr-v151-cen22a, min2022learn}.

There have also been other recent efforts that leverage deep learning for matching (in the context of online bipartite matching~\citep{alomrani2022deep}) and other related combinatorial optimization problems~\citep{BENGIO2021405}. Most of these papers adopt a reinforcement learning based approach to compute their solutions. Our approach, on the other hand, is not sequential but rather end-to-end differentiable, and our parameter weights are updated through a single backward pass. Additionally, the focus of our work is on matching markets and mechanism design, and is concerned with capturing core economic concepts within a machine learning framework and balancing the trade-offs between stability and strategy-proofness.

\section{Preliminaries}
\label{sec:prelim}

Let $W$ denote a set of $n$ {\em workers} and $F$ denote a set of $m$ {\em firms}. A feasible {\em matching}, $\mu$, is a set of (worker, firm) pairs, with each worker and firm participating in at most one match. Let ${\mathcal B}$ denote the set of all {\em matchings}. If $(w,f)\in \mu$, then $\mu$ matches $w$ to $f$, and we write $\mu(w)=f$ and $\mu(f)=w$. If a worker or firm remains unmatched, we say it is matched to $\bot$. We also write $(w,\bot)\in \mu$ (resp.~$(\bot,f)\in \mu$).
Each worker has a {\em strict preference order}, $\succ_w$, over the set $\overline{F} = F\cup \{\bot\}$. Each firm has a strict preference order, $\succ_f$, over the set $\overline{W} = W\cup \{\bot\}$.
Worker $w$ (firm $f$) prefers remaining unmatched to being matched with a firm (worker) ranked below $\bot$ (the agents ranked below $\bot$ are said to be {\em unacceptable}).
If worker $w$ prefers firm $f$ to $f'$, then we write $f \succ_w f'$, similarly for a firm's preferences. Let $P$ denote the set of all {\em preference profiles}, with $\succ=(\succ_1,\ldots,\succ_n,\succ_{n+1},\succ_{n+m})\in P$ denoting a preference profile comprising of the preference order of the $n$ workers and then the $m$ firms.

A pair $(w,f)$ forms a {\em blocking pair for matching $\mu$} if $w$ and $f$ prefer each other to their partners in $\mu$ (or $\bot$ in the case that one or both are unmatched). A matching $\mu$ is {\em stable} if and only if there are no blocking pairs.
A matching $\mu$ is {\em individually rational} (IR) if and only if it is not blocked by any individual; i.e., no agent finds its match unacceptable and prefers $\bot$.\footnote{Stability precludes empty matchings. For example, if a matching $\mu$ leaves a worker $w$ and a firm $f$ unmatched, where $w$ finds $f$ acceptable, and $f$ finds $w$ acceptable, then $(w,f)$ is a blocking pair to $\mu$.}

\subsection{Randomized matchings. } We work with {\em randomized matching mechanisms}, $g$, that map preference profiles, $\succ$, to distributions on matchings, denoted $g(\succ)\in \triangle({\mathcal B})$ (the probability simplex on matchings).
Let $r\in [0,1]^{(n+1)\times (m+1)}$ denote the {\em marginal probability}, $r_{wf}\geq 0$, with which worker $w$ is matched with firm $f$, for each $w\in \overline{W}$ and $f\in \overline{F}$. We require $\sum_{f'\in \overline{F}} r_{wf'}= 1$ for all $w\in W$, and $\sum_{w'\in \overline{W}}r_{w'f}=1$ for all $f\in F$. For notational simplicity, we write $g_{wf}(\succ)$ for the marginal probability of matching worker $w$ (or $\bot$) and firm $f$ (or $\bot$). 
\begin{theorem}[Birkhoff von-Neumann]
Given any randomized matching $r$, there exists a distribution on matchings, $\Delta({\mathcal B})$,
with marginal probabilities equal to $r$. 
\end{theorem}
The following definition is standard~\citep{BCKM13}, and generalizes stability to randomized matchings.
\begin{definition}[Ex ante justified envy]
\label{def:envy}
A randomized matching $r$ causes {\em ex ante justified envy} if: \newline
\hspace*{10pt} (1) some worker $w$ prefers $f$ over some fractionally matched firm $f'$ (including $f'=\bot$) and firm $f$ prefers $w$ over some fractionally matched worker $w'$ (including $w'=\bot$) (``$w$ has envy towards $w'$" and ``$f$ has envy towards $f'$"), or \newline
\hspace*{10pt} (2) some worker $w$ finds a fractionally matched $f'\in F$ unacceptable, i.e. $r_{wf'}>0$ and $\bot \succ_w f'$, or some firm $f$ finds a fractionally matched $w'\in W$ unacceptable, i.e. $r_{w'f}>0$ and $\bot \succ_f w'$.
\end{definition}

A randomized matching $r$ is {\em ex ante stable} if and only if it does not cause any {\em ex ante} justified envy. Ex ante stability reduces to the standard concept of stability for deterministic matching.
Part~(1) of the definition includes non-wastefulness: for any worker $w$, we should have $r_{w\bot} = 0$ if there exists some firm $f' \in F$ for which $r_{w'f'} > 0$ and $w \succ_{f'} w'$ and for for any firm $f$, we need $r_{\bot f} = 0$ if there exists some worker $w' \in W$ for which $r_{w'f'} > 0$ and $f \succ_{w'} f'$. Part~(2) of the definition captures IR: for any worker $w$, we should have $r_{wf'} = 0$ for all $f'\in F$ for which $\bot \succ_w f'$, and for any firm $f$, we need $r_{w'f} = 0$ for all $w'\in W$ for which $\bot \succ_f w'$. 

To define strategy-proofness, say that $u_w: \overline{F}\to \mathbb{R}$ is a {\em $\succ_w$-utility} for worker $w$ when $u_w(f)>u_w(f')$ if and only if $f\succ_w f'$, for all $f, f'\in \overline{F}$. We similarly define a $\succ_f$-utility for a firm $f$.
The following concept of ordinal SP is standard~\citep{Erdil14}, and generalizes SP to randomized matchings.
\begin{definition}[Ordinal strategy-proofness]\label{def:ordinal-sp}
A randomized matching mechanism $g$ satisfies
{\em ordinal SP} if and only if, for all agents $i\in W\cup F$, for any preference profile $\succ$, and any $\succ_i$-utility for agent $i$, 
and for all reports $\succ'_i$, we have
\begin{align}
 \mathbf{E}_{\mu\sim g(\succ_i, \succ_{-i})}[u_i(\mu(i))] \geq \mathbf{E}_{\mu\sim g(\succ'_i, \succ_{-i})}[u_i(\mu(i))]. 
\end{align} 
\end{definition}

By this definition, no worker or firm can improve their expected utility (for any utility function consistent with their preference order) by misreporting their preference order. For a deterministic mechanism, ordinal SP reduces to standard SP. \citet{Erdil14} shows that {\em first-order stochastic dominance} is equivalent to ordinal SP.
\begin{definition}[First Order Stochastic Dominance]\label{def:fosd-sp}
 A randomized matching mechanism $g$ satisfies
 {\em first order stochastic dominance} (FOSD) if and only if, for worker $w$, and each $f'\in \overline{F}$ such that $f' \succ_{w} \bot$,
 and all reports of others $\succ_{-w}$, we have (and similarly for the roles of workers and firms transposed),
 \begin{align}
\!\!\sum_{f\in F: f\succ_w f'} \!\!\!g_{wf}(\succ_w, \succ_{-w}) &\! \geq \!\!\!\sum_{f\in F: f\succ_w f'} \!\!g_{wf}(\succ'_w, \succ_{-w}).\! 
\end{align}
\end{definition}

FOSD states that, whether looking at its most preferred firm, its two most preferred firms, or so forth, worker $w$ achieves a higher probability of matching on that set of firms for its true report than for any misreport. We make use of a quantification of the violation of this condition to provide a surrogate for the failure of SP during learning.

\begin{theorem}[\citep{Erdil14}]\label{thm:fosd}
A two-sided matching mechanism is ordinal SP if and only if it satisfies FOSD.
\end{theorem}

\subsection{Deferred Acceptance, RSD, and TTC.}
We consider three benchmark mechanisms: the stable but not SP \emph{deferred-acceptance} (DA) mechanism, the SP but not stable \emph{randomized serial dictatorship} (RSD) mechanism, and the \emph{Top Trading Cycles} (TTC) mechanism, which is neither SP nor stable.
The DA and TTC mechanisms are ordinal SP for the proposing side of the market but not for agents on both sides of the market.

\begin{definition}[Deferred-acceptance (DA)]\label{def:DA}
In {\em worker-proposing deferred-acceptance} (firm-proposing is defined analogously), each worker $w$ maintains a list of acceptable firms ($f\succ_w \bot$) for which it has not had a proposal rejected (``remaining firms"). Repeat until all proposals are accepted: 
\begin{itemize}
\item $\forall w\in W$: $w$ proposes to its best acceptable, remaining firm.
\item $\forall f\in F$: $f$ tentatively accepts its best proposal (if any), and rejects the rest.
\item$\forall w\in W$: If $w$ is rejected by firm $f$, it updates its list of acceptable firms to remove $f$.
\end{itemize}
\end{definition}

\begin{theorem}[see \citep{RothSotomayor1990}]
DA is stable but not Ordinal SP.
\end{theorem}

\begin{definition}[Randomized serial dictatorship (RSD)]\label{def:RSD}
In the two-sided version of RSD, we first sample a {\em priority order}, $\pi$, on the set $W\cup F$, uniformly at random, such that $\pi = (\pi_1, \pi_2, \ldots, \pi_{m+n})$ is a permutation on $W\cup F$ in decreasing order of priority. For the one-sided version, we sample a priority order $\pi$ on either $W$ or $F$.

Proceed as follows:
\begin{itemize}
\item Initialize matching $\mu$ to the empty matching.
\item In round $k=1, \ldots, \vert\pi\vert$:
\begin{itemize}
 \item[--] If $\pi_k$ is not yet matched in $\mu$, then add to
 $\mu$ the match between $\pi_k$ and its most preferred
 unmatched agent, or $\bot$ if all
 remaining agents are unacceptable to $\pi_k$.
\end{itemize}
\end{itemize}
\end{definition}

\begin{theorem}\label{thm:rsd-unstable}
RSD satisfies FOSD---and thus is ordinal SP by Theorem~\ref{thm:fosd}---but is not stable. 
\end{theorem}
\begin{proof}
We defer the proof to appendix~\ref{app:rsd-unstable}
\end{proof}

\begin{definition}[Top Trading Cycles (TTC)]\label{def:TTC}
In {\em worker-proposing TTC} (firm-proposing is defined analogously), each agent (worker or firm) maintains a list of acceptable firms. Repeat until all agents are matched:
\begin{itemize}
 \item Form a directed graph with each unmatched agent pointing to their most preferred option. The agents can point at themselves if there are no acceptable options available. Every worker that is a part of a cycle is matched to a firm it points to ( or itself, if the worker is pointing at itself). The unmatched agents remove from their lists every matched agent from this round.
\end{itemize}
\end{definition} 

\begin{theorem}\label{thm:ttc-unstable}
TTC is neither strategy-proof nor stable for both sides.
\end{theorem}
\begin{proof}
Example~\ref{ex:2} in Appendix~\ref{app:ttc-unstable} shows that TTC is neither strategyproof nor stable.
\end{proof}

\begin{remark} 
TTC, like RSD, is usually used in one-sided assignment problems, where it is SP, and where the notion of stability which is an important consideration in two-sided matching, is not a concern.
\end{remark}

\section{Two-Sided Matching as a Learning Problem}\label{sec:learning}

In this section, we develop the use of deep learning for the design of two-sided matching mechanisms. 
\subsection{Neural Network Architecture}\label{sec:matching-methodology}

\begin{figure*}[t]
\centering
\begin{tikzpicture}[scale=0.64, transform shape, shorten >=1pt,->,draw=black!100, node distance=\layersep, thick]
    \tikzstyle{input text}=[draw=white,minimum size=22pt,inner sep=0pt]
    \tikzstyle{input neuron}=[circle,draw=black!100,minimum size=17pt,inner sep=0pt,thick]
    \tikzstyle{hidden neuron}=[circle,draw=black!100,minimum size=25pt,inner sep=0pt,thick]
    \tikzstyle{hidden text}=[draw=white,minimum size=22pt,inner sep=0pt,thick]
    \tikzstyle{unit}=[circle,draw=black!100,minimum size=25pt,inner sep=0pt,thick]
    \tikzstyle{squnit}=[draw=black!100,minimum size=20pt,inner sep=0pt,thick]
     
    \def\x{0.25}
    \node[input text] (I-10) at (0,-0.65) {$\vdots$};
    \node[input text] (I-20) at (0,-4.2) {$\vdots$};
    \node[input neuron, pin={[pin edge={<-}]left:$p_{11}$}] (I-1) at (0,0) {};
    \node[input neuron, pin={[pin edge={<-}]left:$p_{nm}$}] (I-2) at (0,-1.5) {};
    \node[input neuron, pin={[pin edge={<-}]left:$q_{11}$}] (I-3) at (0,-3.5) {};
	\node[input neuron, pin={[pin edge={<-}]left:$q_{nm}$}] (I-4) at (0,-5) {};  
	
	\node[hidden text] (temp-1) at (-0.6, 0.75) {$\beta$};
    \node[hidden text] (temp-1) at (-0.6, -0.75) {$\beta$};
    \node[hidden text] (temp-1) at (-0.6, -2.75) {$\beta$};
    \node[hidden text] (temp-1) at (-0.6, -4.25) {$\beta$};
    
	\draw[<-,solid,line width=0.5mm,dotted,color=gray] (-0.05, 0.25)  -- (-0.55, 0.75 + 0.05); 
	\draw[<-,solid,line width=0.5mm,dotted,color=gray] (-0.05, -1.25)  -- (-0.55, -0.75 + 0.05); 
	\draw[<-,solid,line width=0.5mm,dotted,color=gray] (-0.05, -3.25)  -- (-0.55, -2.75 + 0.05); 
	\draw[<-,solid,line width=0.5mm,dotted,color=gray] (-0.05, -4.75)  -- (-0.55, -4.25 + 0.05);

	\node[input text] (temp-1111) at (2.25, -2.5) {$\ldots$};

    \path node[hidden text] (H-0) at (1.5,-3.25) {\vdots};
    \path node[hidden neuron] (H-1) at (1.5,-0.5) {$h^{(1)}_1$};
    \path node[hidden neuron] (H-2) at (1.5,-2) {$h^{(1)}_2$};
    \path node[hidden neuron] (H-3) at (1.5,-4.5) {$h^{(1)}_{J_1}$};

    \foreach \dest in {1,...,3}
    	\foreach \src in {1,...,4}
			\path (I-\src) edge (H-\dest); 		
    
    \path node[hidden text] (J-0) at (3,-3.25) {\vdots};
    \path node[hidden neuron] (J-1) at (3,-0.5) {$h^{(R)}_1$};
    \path node[hidden neuron] (J-2) at (3,-2) {$h^{(R)}_2$};
    \path node[hidden neuron] (J-3) at (3,-4.5) {$h^{(R)}_{J_{R}}$};

    \node[input text] (S-10) at (4.55,-0.65) {$\vdots$};
    \node[input text] (S-20) at (4.55,-4.2) {$\vdots$};
    \path node[unit] (S-1) at (4.55,-0) {$s_{11}$};
    \path node[unit] (S-2) at (4.55,-1.5) {$s_{n+1m}$};
    \path node[unit] (S-3) at (4.55,-3.5) {$s'_{11}$};
    \path node[unit] (S-4) at (4.55,-5.0) {$s'_{nm+1}$};

     \foreach \src in {1,...,3}
		\foreach \dest in {1,...,4}
			\path (J-\src) edge (S-\dest);

    \path node[squnit] (T-1) at (5.75,-0) {$\sigma_+$};
    \path node[squnit] (T-2) at (5.75,-1.5) {$\sigma_+$};
    \path node[squnit] (T-3) at (5.75,-3.5) {$\sigma_+$};
    \path node[squnit] (T-4) at (5.75,-5.0) {$\sigma_+$};
    
    \node[input text] (T-10) at (5.75,-0.65) {$\vdots$};
    \node[input text] (T-20) at (5.75,-4.2) {$\vdots$};
    
    \foreach \src in {1,...,4}
			\path (S-\src) edge (T-\src); 
			
	\path node[input neuron] (U-1) at (\x +7.0,0.5-0) {$\times$};
    \path node[input neuron] (U-2) at (\x +7.0,-0.25-1.5) {$\times$};
    \path node[input neuron] (U-3) at (\x +7.0,-0.25-3.5) {$\times$};
    \path node[input neuron] (U-4) at (\x +7.0,-0.25-5.0) {$\times$};
    
    \node[hidden text] (temp-1) at (6.5, 1.25) {$\beta$};
    \node[hidden text] (temp-1) at (6.5, -1.0) {$\beta$};
    \node[hidden text] (temp-1) at (6.5, -3.0) {$\beta$};
    \node[hidden text] (temp-1) at (6.5, -4.5) {$\beta$};
    
	\draw[<-,solid,line width=0.5mm,dotted,color=gray] (7.05, 0.75)  -- (6.55, 1.25 + 0.05); 
	\draw[<-,solid,line width=0.5mm,dotted,color=gray] (7.05, -1.5)  -- (6.55, -1.0 + 0.05); 
	\draw[<-,solid,line width=0.5mm,dotted,color=gray] (7.05, -3.5)  -- (6.55, -3.0 + 0.05); 
	\draw[<-,solid,line width=0.5mm,dotted,color=gray] (7.05, -5.0)  -- (6.55, -4.5 + 0.05);
    
     \foreach \src in {1,...,4}
			\path (T-\src) edge (U-\src);
    
    \path node[hidden neuron] (V-1) at (\x +8.25,0.5-0) {$\bar s_{11}$};
    \path node[hidden neuron] (V-2) at (\x +8.25,0.5-1.5) {$\bar s_{n1}$};
    \path node[hidden neuron] (V-25) at (\x +8.25,0.5-2.75) {$\bar s_{\bot 1}$};
    \path node[hidden neuron] (V-3) at (\x +8.25,0.25-4.0) {$\bar s'_{11}$};
    \path node[hidden neuron] (V-4) at (\x +8.25,0.25-5.5) {$\bar s'_{n1}$};

	\path (U-1) edge (V-1);
	\path (U-2) edge (V-25);
	\path (U-3) edge (V-3);
	\path (U-4) edge (V-4);

	\path node[hidden neuron] (W-1) at (\x +10.25,0.5-0) {$\bar s_{1m}$};
    \path node[hidden neuron] (W-2) at (\x +10.25,0.5-1.5) {$\bar s_{nm}$};
    \path node[hidden neuron] (W-25) at (\x +10.25,0.5-2.75) {$\bar s_{\bot m}$};
    \path node[hidden neuron] (W-3) at (\x +10.25,0.25-4.0) {$\bar s'_{1m}$};
    \path node[hidden neuron] (W-4) at (\x +10.25,0.25-5.5) {$\bar s'_{nm}$};
    \path node[hidden neuron] (W-35) at (\x +11.55,0.25-4.0) {$\bar s'_{1\bot}$};
    \path node[hidden neuron] (W-45) at (\x +11.55,0.25-5.5) {$\bar s'_{n\bot}$};

    
    \node[input text] (W-11) at (\x +9.25,0.5-0) {$\ldots$};
    \node[input text] (W-12) at (\x +9.25,0.5-1.5) {$\ldots$};
    \node[input text] (W-13) at (\x +9.25,-0.25-3.5) {$\ldots$};
    \node[input text] (W-14) at (\x +9.25,-0.25-5) {$\ldots$};
    \node[input text] (W-15) at (\x +9.25,-1.5-0.65) {$\vdots$};
    \node[input text] (W-16) at (\x +9.25,-0.25-4.2) {$\vdots$};

    \node[input text] (temp-111) at (\x +12.5+1.0,-0.65) {$\vdots$};
    \node[input text] (temp-222) at (\x +12.5+1.0,-2.42) {$\vdots$};
    \node[input text] (temp-333) at (\x +12.5+1.0,-4.2) {$\vdots$};
    \path node[unit, pin={[pin edge={->}]right:$r_{11} = \min\{\hat{s}_{11},\hat{s}'_{11}\}$}] (X-1) at (\x +12.5+1.0,-0) {};
    \path node[unit, pin={[pin edge={->}]right:$r_{n1} = \min\{\hat{s}_{n1},\hat{s}'_{n1}\}$}] (X-2) at (\x +12.5+1.0,-1.5) {};
    \path node[unit, pin={[pin edge={->}]right:$r_{1m} = \min\{\hat{s}_{1m},\hat{s}'_{1m}\}$}] (X-3) at (\x +12.5+1.0,-3.5) {};
    \path node[unit, pin={[pin edge={->}]right:$r_{nm} = \min\{\hat{s}_{nm},\hat{s}'_{nm}\}$}] (X-4) at (\x +12.5+1.0,-5) {};

	\draw[thick,dotted] ($(X-1.north west)+(-0.2,0.3)$)  rectangle ($(X-4.south east)+(0.2,-0.3)$);
	\draw[thick,dotted] ($(W-1.north west)+(-0.2,0.3)$)  rectangle ($(W-25.south east)+(0.2,-0.3)$);
	\draw[thick,dotted] ($(V-1.north west)+(-0.2,0.3)$)  rectangle ($(V-25.south east)+(0.2,-0.3)$);
	\draw[thick,dotted] ($(V-3.north west)+(-0.2,0.3)$)  rectangle ($(W-35.south east)+(0.2,-0.3)$);
	\draw[thick,dotted] ($(V-4.north west)+(-0.2,0.3)$)  rectangle ($(W-45.south east)+(0.2,-0.3)$);
	
	\draw[thick,dashed] ($(V-1.north west)+(-0.3,0.4)$)  rectangle ($(W-2.south east)+(0.3,-0.4)$);
	\draw[thick,dashed] ($(V-3.north west)+(-0.3,0.4)$)  rectangle ($(W-4.south east)+(0.3,-0.4)$);
	
	\node at  ($(W-16.west) + (0.3, -0.05-1.75)$) {$row-wise$};
	\node at  ($(W-16.west) + (0.3, -0.05-2.0)$) {$normalization$};
	
	\node at  ($(V-1.east) + (0.3, +1.25)$) {$column-wise$};
	\node at  ($(V-1.east) + (0.3, +1.0)$) {$normalization$};
	
    
	
	\coordinate (LL1) at (\x +10.70,-0.35);
    \coordinate (LL2) at (\x +12.5+0.5,-1.65);
    
    \coordinate (LL3) at (\x +10.70,-3.35);
    \coordinate (LL4) at (\x +12.5+0.5,-2.30);
    
    \node[hidden text] (LL-5) at (12.5, -1.0) {$\hat{s}$};
    \node[hidden text] (LL-5) at (12.5, -2.6) {$\hat{s}'$};

	\draw [thick,black] (LL1) to[out=0,in=180] (LL2);
	\draw [thick,black] (LL3) to[out=45,in=180] (LL4);

\end{tikzpicture}
%
\caption{Matching network $g$ for a set of $n$ workers and $m$ firms. Given inputs $p, q \in \mathbb{R}^{n \times m}$ the matching network is a feed-forward network with $R$ hidden layers that uses {\em softplus} activation to generate non-negative scores and normalization to compute the randomized matching. We additionally generate a Boolean mask matrix, $\beta$, and multiply it with the score matrix before normalization to ensure IR by making the probability of matches that are unacceptable zero.}
\label{fig:gen-net-2}
\end{figure*}

We use a neural network to represent a matching mechanism. Let $g^\theta: P\to \triangle({\mathcal B})$ denote the mechanism, for parameters $\theta\in \mathbb{R}^d$.
The input is a preference profile, and the output defines a distribution on matchings.
We use a feed-forward neural network with $R = 4$ fully connected hidden layers, $J = 256$ units in each layer, {\em leaky ReLU} activations, and a fully connected output layer.\footnote{Leaky ReLU is a variation of ReLU that addresses the issue of dead neurons by allowing a small positive gradient for negative input values, thereby preventing gradients from getting stuck at $0$}. See Figure~\ref{fig:gen-net-2}.

To represent an agent's preference order in the input, we adopt a 
utility for each agent on the other side of the
market that has a constant offset in utility across successive agents in the preference order. This is purely a representation choice and does not imply that we use this particular utility to study SP (on the contrary, we work with a FOSD-based quantification of the degree of approximation to ordinal SP).
In particular, let $p^\succ_w = (p_{w1}^\succ, \ldots, p_{wm}^\succ)$ and $q^\succ_f= (q_{1f}^\succ, \ldots, q_{nf}^\succ)$
represent the preference order of a worker and firm, respectively. 
We define $p^{\succ}_{w\bot} = 0$ and $q^\succ_{\bot f}=0$, and to further illustrate this representation: 
\begin{itemize}
\item For a preference order $\succ$ with $f_1: w_2, w_1, w_3$, we represent this at the input as $q^\succ_{f_1}=(\frac{2}{3},1,\frac{1}{3})$. 
\item For a preference order $\succ$ with $w_1: f_1, f_2, \bot, f_3$, we represent this as $p^\succ_{w_1}=(\frac{2}{3},\frac{1}{3},-\frac{1}{3})$. 
\item For a preference order $\succ$ with $w_1: f_1, f_2, \bot, f_3, f_4$, we represent this as $p^\succ_{w_1}=(\frac{2}{4},\frac{1}{4},-\frac{1}{4},-\frac{2}{4})$. 
\end{itemize}
Formally, we have:
\begin{align}
 p_{wj}^\succ &= \frac{1}{m}\left( \mathbf{1}_{ j \succ_{w} \bot} + \sum_{j'=1}^{m} \left(\mathbf{1}_{j \succ_{w} j'} - \mathbf{1}_{\bot \succ_{w} j'}\right) \right) \\
 q_{if}^\succ &= \frac{1}{n}\left( \mathbf{1}_{i \succ_{f} \bot} + \sum_{i'=1}^{n} \left(\mathbf{1}_{i \succ_{f} i'} - \mathbf{1}_{\bot \succ_{f} i}\right)\right)
\end{align}
where $\mathbf{1}_X$ is the indicator function for event $X$.

Taken together, the input is vector $(p_{11}^\succ, \ldots, p_{nm}^\succ, q_{11}^\succ, \ldots, q_{nm}^\succ)$ of 
$2\times n \times m$ numbers.

The output of the network is a vector $r \in [0, 1]^{n \times m}$, with $\sum_{j = 1}^{m} r_{wj} \leq 1$ and 
$\sum_{i = 1}^{n} r_{if} \leq 1$ for every every $w \in [n]$ and $f \in [m]$. 
This defines the marginal probabilities in a randomized matching for this input profile. 
The network first outputs two sets of scores $s\in \mathbb{R}^{(n + 1) \times m}$ and $s' \in \mathbb{R}^{n \times (m + 1)}$.
We apply the {\em softplus function} (denoted by $\sigma_+$) element-wise to these scores, where $\sigma_+(x) = \ln (1 + e^{x})$. 
To ensure IR, we first construct a Boolean mask variable $\beta_{wf}$, which is zero only when the match is unacceptable to one or both the worker and firm, i.e., when $\bot \succ_{w} f$ or $\bot \succ_{f} w$. We set $\beta_{n+1, f} = 1$ for $f \in F$ and $\beta_{w, m+1} = 1$ for $w \in W$. We multiply the scores $s$ and $s'$ element-wise with the corresponding Boolean mask variable to compute $\bar s \in \mathbb{R}_{\geq 0}^{(n + 1) \times m}$ and $\bar s' \in \mathbb{R}_{\geq 0}^{n \times (m + 1)}$.

For each $w \in \overline{W}$, we have $\bar s_{wf} = \beta_{wf} \ln (1 + e^{s_{wf}})$, for all $f\in F$. For each $f \in \overline{F}$, we have $\bar s'_{wf} = \beta_{wf} \ln (1 + e^{s'_{wf}})$, for all $w \in W$.
We normalize $\bar s$ along the rows and $\bar s'$ along the columns to obtain {\em normalized scores}, $\hat s$ and $\hat s'$ respectively. The match probability $r_{wf}$, for worker $w \in W$ and firm $f \in F$, is computed as the minimum of the normalized scores:
$r_{wf} = \min\left(\frac{\bar s_{wf}}{\sum_{f' \in \overline{F}} \bar s_{wf'}}, 
\frac{\bar s'_{wf}}{\sum_{w' \in \overline{W}} \bar s'_{w'f}} \right)$.
We have $r_{wf} = 0$ whenever $\beta_{wf} = 0$, ensuring that every matching 
in the support of the distribution will be IR. 
Based on our construction, the allocation matrix $r$ is weakly doubly stochastic, with rows and columns summing to at most 1.
\citet{BCKM13} show that any weakly doubly stochastic matrix can be decomposed to a convex combination of 0-1, weakly doubly stochastic matrices.

\subsection{Formulation as a Learning Problem}
We formulate a {\em loss function} $\mathcal{L}$ that is defined on training data of $\ell$ preference profiles, $D=\{\succ^{(1)},\ldots,\succ^{(\ell)}\}$. Each preference profile $\succ$ sampled i.i.d.~from a distribution on profiles. 
We allow for {\em correlated preferences}; i.e., workers may tend to agree that one of the firms is preferable to one of the other firms, and similarly for firms. 
The loss function captures a tradeoff between stability and ordinal SP.
Recall that $g^\theta(\succ) \in [0, 1]^{n \times m}$ denotes the randomized matching. We write $g^\theta_{w\bot}(\succ) = 1 - \sum_{f = 1}^{m} g^\theta_{wf}(\succ)$ and $g^\theta_{\bot f}{(\succ)} = 1 - \sum_{w=1}^{n} g^\theta_{wf}(\succ)$ to denote the probability of worker $w$ and firm $f$ being unmatched, respectively.

\noindent \textbf{Stability Violation.~~}
For worker $w$ and firm $f$, we define the {\em stability violation }
at profile $\succ$ as
\begin{align}
\mathit{stv}_{wf}(&g^\theta, \succ) \!=\! 
\left(\sum_{w' \in \overline{W}} g^\theta_{w'f}(\succ)\cdot\max\{q_{wf}^\succ - q_{w'f}^\succ, 0\} \right)\notag \\
 & \times \! \left(\sum_{f' \in \overline{F}} g^\theta_{wf'}(\succ)\cdot \max\{p_{wf}^\succ - p_{wf'}^\succ,0\}\!\!\right). 
\end{align}

This captures the first kind of {\em ex ante} justified envy in Definition~\ref{def:envy}.
We can omit the second kind of {\em ex ante} justified envy because the learned mechanisms satisfy IR through the use of the Boolean mask matrix (and thus, there are no violations of the second kind).

The average stability violation (or just {\em stability violation}) of mechanism $g^\theta$ on profile $\succ$
 is $\mathit{stv}(g^\theta, \succ) = \frac{1}{2}\left(\frac{1}{m} +\frac{1}{n}\right)\sum_{w = 1}^{n}\sum_{f = 1}^m \mathit{stv}_{wf}{(g^\theta, \succ)}.$
We define the {\em expected stability violation}, $\mathit{STV}(g^\theta)=\mathbb{E}_{\succ} \mathit{stv}(g^\theta,\succ)$. 
We also write $\mathit{stv}(g^\theta)$ to denote the average stability violation on the training data.
\begin{theorem}\label{thm:stability-violation}
A randomized matching mechanism $g^\theta$ is ex ante stable up to zero-measure events if and only if $\mathit{STV}(g^\theta)=0$. 
\end{theorem}
\begin{proof}
We defer the proof to appendix~\ref{app:stv-proof}
\end{proof}

\noindent \textbf{Ordinal SP violation.~~}
We turn now to quantifying the degree of approximation to ordinal SP. Let $\succ_{-i}=(\succ_1,\ldots,\succ_{i-1},\succ_{i+1},\ldots,\succ_{n + m})$. For a valuation profile, $\succ\ \in P$, and a mechanism $g^\theta$, let $\Delta_{wf}(g^\theta, \succ'_{w}, \succ) = g^\theta_{wf}(\succ'_w, \succ_{-w}) - g^\theta_{wf}(\succ_w, \succ_{-w})$. The {\em regret} to worker $w$ (firm $f$) is defined as:
\begin{equation}
\begin{split}
\mathrm{regret}_w(g^\theta, \succ) &= \max_{\succ'_{w} \in P} \left( \max_{f' \succ_{w} \bot} \sum_{f\succ_w f'} \Delta_{wf}(g^\theta, \succ'_{w}, \succ) \right). \\
\mathrm{regret}_f(g^\theta, \succ) &= \max_{\succ'_{f} \in P} \left( \max_{w' \succ_{f} \bot} \sum_{w\succ_f w'} \Delta_{wf}(g^\theta, \succ'_f, \succ) \right). 
\end{split} 
\end{equation}
The {\em average regret} on profile $\succ$ is:
\begin{align}
 \mathit{regret}(g^\theta,\succ)= \frac{1}{2} \left(\frac{1}{m}\sum_{w\in W}\mathrm{regret}_w(g^\theta, \succ)+ \frac{1}{n}\sum_{f\in F} \mathrm{regret}_f(g^\theta, \succ) \right) 
\end{align} The {\em expected regret} is $\mathit{RGT}(g^\theta)=\mathbb{E}_\succ \mathit{regret}(g^\theta,\succ)$. We can also write $\mathit{rgt}(g^\theta)$ to denote the average regret 
on training data.

\begin{theorem}\label{thm:regret}
The regret to a worker (firm) for a given preference profile
is the maximum amount by which the worker (firm) can increase their expected normalized utility through a misreport, fixing the reports of others. 
\end{theorem}
\begin{proof}
Consider some worker $w \in W$. Without loss of generality, let $\succ_w: f_1, \ldots, f_k, \bot, f_{k+1},\ldots, f_m$. Any normalized {\em $\succ_w$-utility~} function, $u_w$, consistent with ordering given by $\succ_w$ satisfies $1 \geq u_w(f_1) \geq u_w(f_2) \geq\ldots u_w(f_k) \geq 0 \geq u_w(f_{k+1}) \geq \ldots u_w(f_m)$. Let $U_w$ be the set of all such consistent utility functions.

Consider some misreport $\succ'_w$. We have $\Delta_{wf}(g^\theta, \succ'_{w}, \succ) = g^\theta_{wf}(\succ'_w, \succ_{-w}) - g^\theta_{wf}(\succ_w, \succ_{-w})$. The increase in utility for worker $w$ when the utility function is $u_w$ is given by $\sum_{f \in F} u_w(f) \Delta_{wf}(g^\theta, \succ'_{w}, \succ) $. The maximum amount by which worker $w$ can increase their expected normalized utility through misreport $\succ_w'$ is given by the objective: $\max_{u_w \in U_w} \sum_{f \in F} u_w(f) \Delta_{wf}(g^\theta, \succ'_{w}, \succ)$. 

Since $g^\theta$ always guarantees IR, we have:
\begin{equation}
    \sum_{f \in F: \bot \succ f} u_w(f) \Delta_{wf}(g^\theta, \succ'_{w}, \succ) = \sum_{f \in F: \bot \succ_w f} u_w(f) g^\theta_{wf}(\succ'_w, \succ_{-w}) \leq 0
\end{equation}. 
Thus, we can simplify our search space by only considering $u_w \in U_w$ where $u_w(f_{k+1}), \ldots, u_w(f_m) = 0$. 

Define $\delta_k = u_w(f_k), \delta_{k-1} = u_w(f_{k-1}) - u_w(f_k), \delta_1 = u_w(f_1) - u_w(f_2)$. This objective can thus be rewritten as:
\begin{align}
 \text{max} &\sum_{f = 1}^{k} \left(\sum_{i = f}^{k} \delta_i \right) \Delta_{wf}(g^\theta, \succ'_{w}, \succ)\\
 \text{such that} &\sum_{i = 1}^{k} \delta_i \leq 1 \text{~and~}
 \delta_1, \ldots \delta_k \geq 0
\end{align}

Changing the order of summation, we have the following optimization problem:
\begin{align}
 \text{max} \sum_{i = 1}^{k} \delta_i & \left(\sum_{f = 1}^{i} \Delta_{wf}(g^\theta, \succ'_{w}, \succ)\right) \\
 \text{such that} \sum_{i = 1}^{k} \delta_i &\leq 1 \text{~and~}\delta_1, \ldots \delta_k \geq 0
\end{align}

This objective is of the form $ \max_{\|x\|_1 \leq 1} x^T y$ and it's solution is given by the $||y||_{\infty}$. Thus, the solution to the above maximization problem is given by $\max_{i \in [k]} \sum_{f=1}^{i} \Delta_{wf}(g^\theta, \succ'_{w}, \succ)$. But this is the same as $\max_{f': f' \succ_w \bot} \sum_{f: f \succ_w f'} \Delta_{wf}(g^\theta, \succ'_{w}, \succ)$. Computing the maximum possible increase over all such misreports gives us $\max_{\succ_{w'} \in P} \left(\max_{f': f' \succ_w \bot} \sum_{f: f \succ_w f'} \Delta_{wf}(g^\theta, \succ'_{w}, \succ)\right)$. This quantity is exactly $\mathit{regret}_w(g^\theta, \succ)$. The proof follows similarly for any firm $f$.
\end{proof}

\begin{theorem}\label{thm:regret0}
A randomized mechanism, $g^\theta$,
 is ordinal SP up to zero-measure events if and only if $\mathit{RGT}(g^\theta)=0$.
\end{theorem}
The proof is similar to Theorem~\ref{thm:stability-violation}, and deferred to Appendix~\ref{app:sp-proof}.

\subsection{Training Procedure}
For a mechanism parameterized as $g^\theta$, the training problem that we formulate is,
 \begin{align}
 &\min_{\theta}\, \lambda \cdot \mathit{stv}(g^\theta) + (1 - \lambda) \cdot \mathit{rgt}(g^\theta),
 \label{eq:match-opt}
 \end{align}
 where $\lambda \in [0, 1]$ 
 controls the tradeoff between approximate stability and approximate SP.
We use SGD to solve~\eqref{eq:match-opt}. The gradient of the degree of violation of stability with respect to network parameters is straightforward to calculate. The gradient of regret is complicated by the nested maximization in the definition of regret. In order to compute the gradient, we first solve the inner maximization by checking
 possible misreports. Let $\hat{\succ}^{(\ell)}_i$ denote the {\em defeating preference report} for agent $i$ (a worker or firm) at preference profile $\succ^{(\ell)}$ that maximizes $regret_i(g^\theta, \succ^{(\ell)})$. Given this, we obtain the derivative of regret for agent $i$ with respect to the network parameters, fixing the misreport to the defeating valuation and adopting truthful reports for the others.

For the case of TTC and RSD, we also need to quantify the IR violation (this is
 necessarily zero for the other mechanisms). We define the IR violation 
at profile $\succ$ as:
\begin{equation}
 \begin{aligned}
 \mathit{irv}(g, \succ) &= \frac{1}{2m}\sum_{w = 1}^{n}\sum_{f = 1}^{m} g_{wf}(\succ)\cdot \left(\max\{-q_{wf}^{\succ}, 0\}\right) \\&+ \frac{1}{2n}\sum_{w = 1}^{n}\sum_{f = 1}^{m} g_{wf}(\succ)\cdot \left(\max\{-p_{wf}^{\succ}, 0\}\right)
 \end{aligned}
\end{equation}

For the case of RSD, we also include the average IR violation on test data when reporting the stability violation.

\section{Experimental Results}\label{sec:experiments}
We report the results on a test set of 204,800 preference profiles, and use the Adam optimizer to train our models for 50,000 mini-batch iterations, with mini-batches of 1024 profiles.
We use the {\em PyTorch} deep learning library, and all experiments are run on a cluster of NVIDIA GPU cores. 
%

We study the following market settings: 
\begin{itemize}
\item For {\em uncorrelated preferences}, for each worker or firm, we sample uniformly at random from all preference orders, and then, with probability 0.2 (truncation probability), we choose at random a position at which to truncate this agent's preference order.
\item For {\em correlated preferences}, we sample a preference profile as in the uncorrelated case. We also sample a common preference order on firms and a common preference order on workers. For each agent, with probability, $p_{\mathrm{corr}}>0$, we replace its preference order with the common preference order for its side of the market.
\end{itemize}

Specifically, we consider matching problems with $n = 4$ workers and $m = 4$ firms with uncorrelated preference and varying probability of correlation $p_{\mathrm{corr}} = \{0.25, 0.5, 0.75\}$. 

\begin{figure*}[t]
\centering
\includegraphics[scale=0.42]{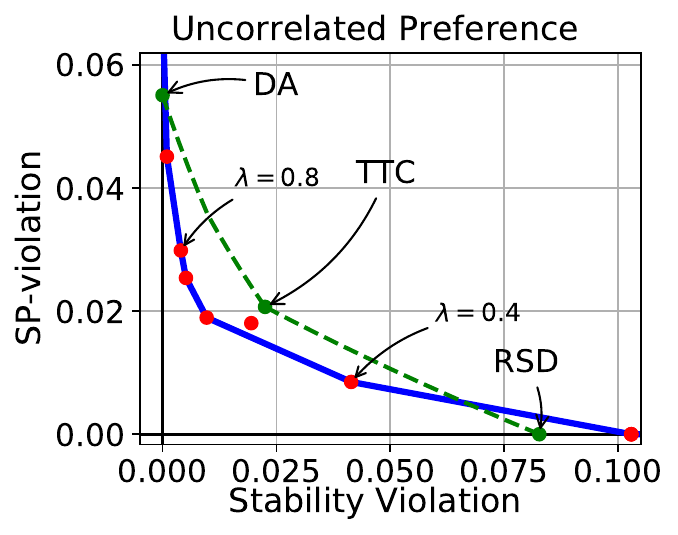}
\includegraphics[scale=0.42]{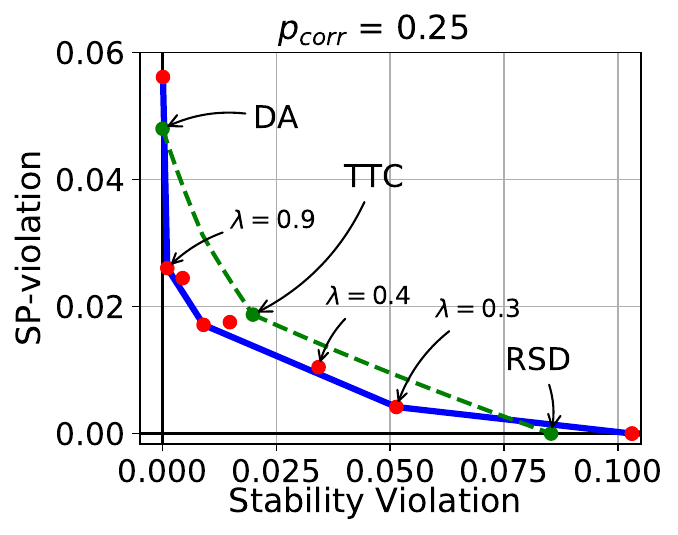}\\
\includegraphics[scale=0.42]{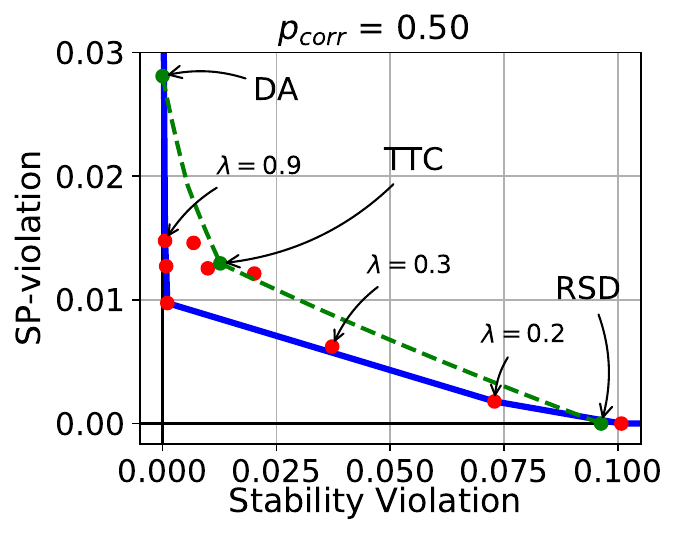}
\includegraphics[scale=0.42]{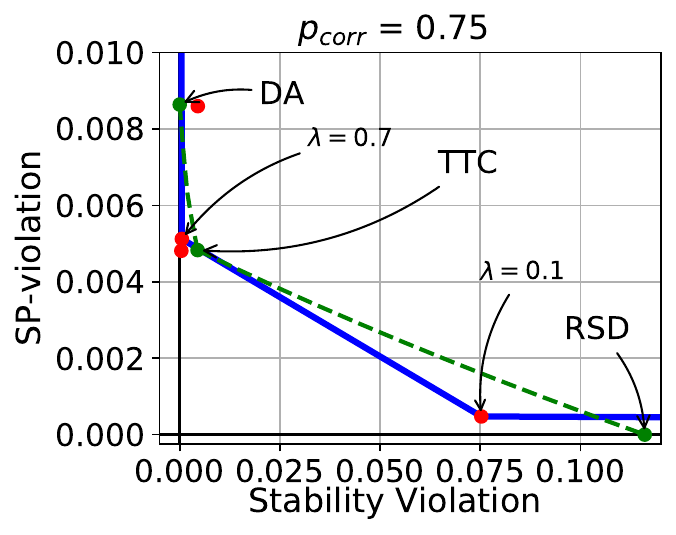}

\caption{\label{fig:frontier} Comparing stability violation and strategy-proofness violation from the learned mechanisms for different choices of $\lambda$ (red dots) with the best of worker- and firm-proposing DA, as well as
 TTC, and RSD, in $4\times4$ two-sided matching, and considering uncorrelated preference orders as well as markets
 with increasing correlation ($p_{\mathrm{corr}}\in \{0.25, 0.5, 0.75\}$). The stability violation for TTC and RSD includes IR violations.}

\end{figure*}

We compare the performance of our mechanisms, varying parameter $\lambda$ between $0$ and $1$, 
 with the best of worker- and firm- proposing DA and TTC (as determined 
by average SP violation over the test data) and RSD\footnote{We only plot the performance of one-sided RSD as it achieves lower stability violation the two-sided version}. We also compare against convex combinations of DA, TTC, and RSD.
We plot the resulting frontier on stability violation ($\mathit{stv}(g^\theta)$) and SP violation ($\mathit{rgt}(g^\theta)$) in Figure~\ref{fig:frontier}. As explained above, because TTC and RSD mechanisms do not guarantee IR, we include the IR violations in the reported stability violation (none of the other mechanisms fail IR).
 
At $\lambda = 0.0$, we learn a mechanism that has very low regret ($\approx0$) but poor stability. This performance is similar to that of RSD. For large values of $\lambda$, we learn a mechanism that approximates DA. For intermediate values, we find solutions that dominate the convex combination of DA, TTC, and RSD and find novel and interesting tradeoffs between SP and stability. Notably, for lower levels of correlations we see substantially better SP than DA along with very little loss in stability. Given the importance of stability in practice, this is a very intriguing discovery. For higher levels of correlations, we see substantially better stability than RSD along with very little loss in SP. It is also interesting to see that TTC itself has intermediate properties, between those of DA and RSD. Comparing the scale of the y-axes, we can also see that increasing correlation tends to reduce the opportunity for strategic behavior across both the DA and the learned mechanisms.

\begin{figure*}[t]
\centering
\includegraphics[scale=0.33]{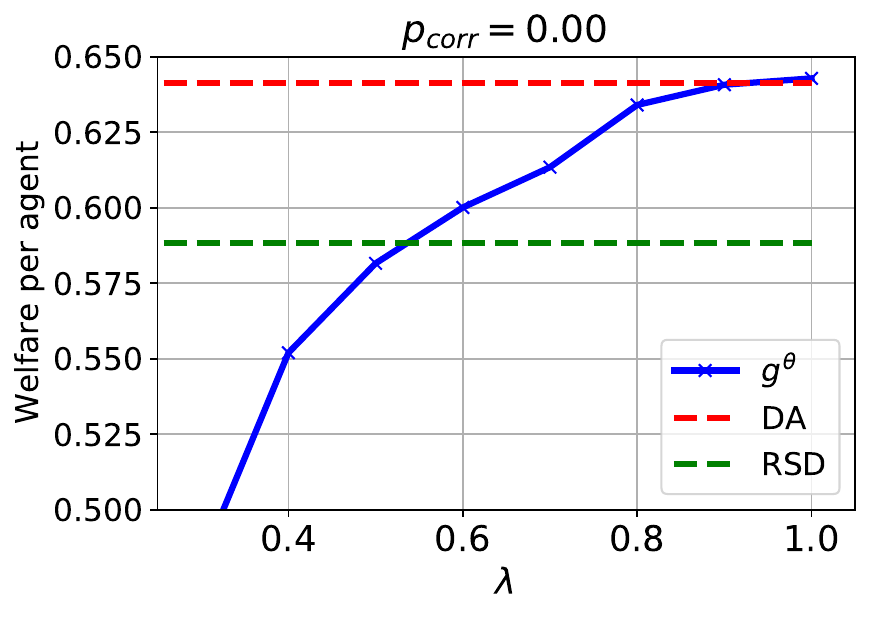}
\includegraphics[scale=0.33]{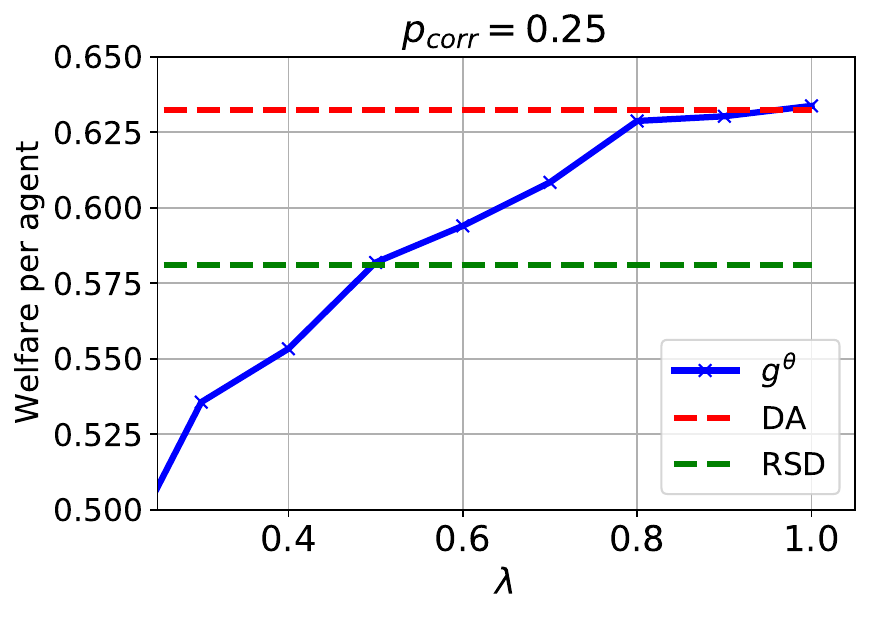}
\includegraphics[scale=0.33]{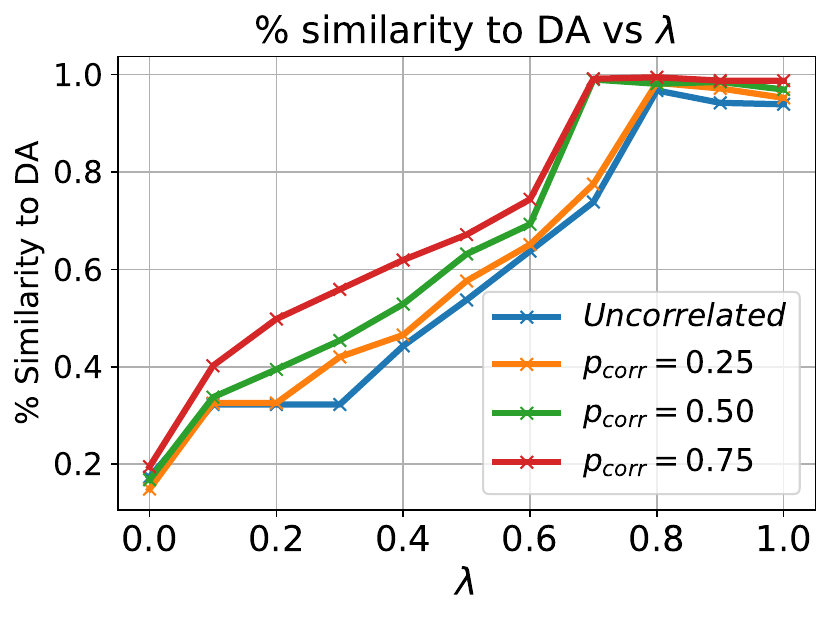}
\includegraphics[scale=0.33]{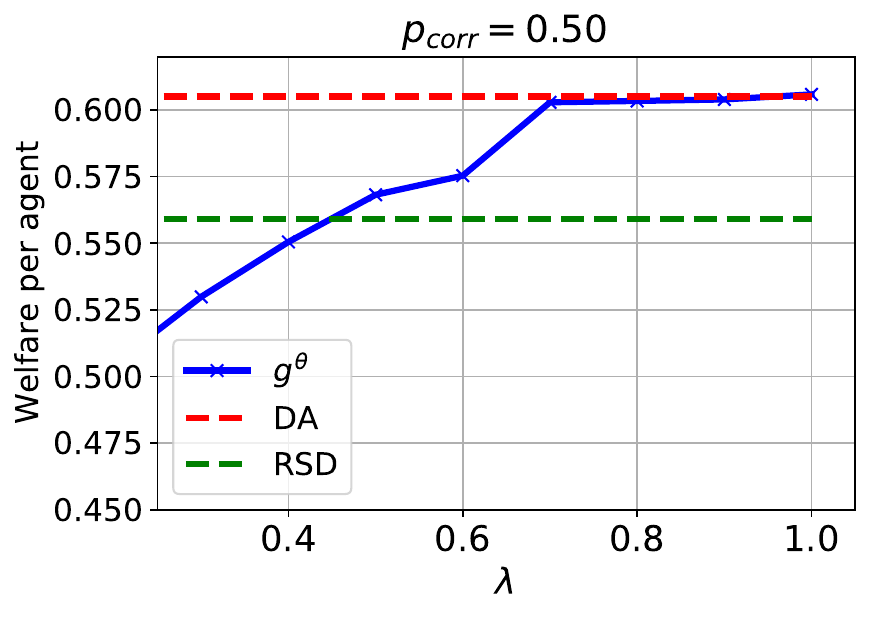}
\includegraphics[scale=0.33]{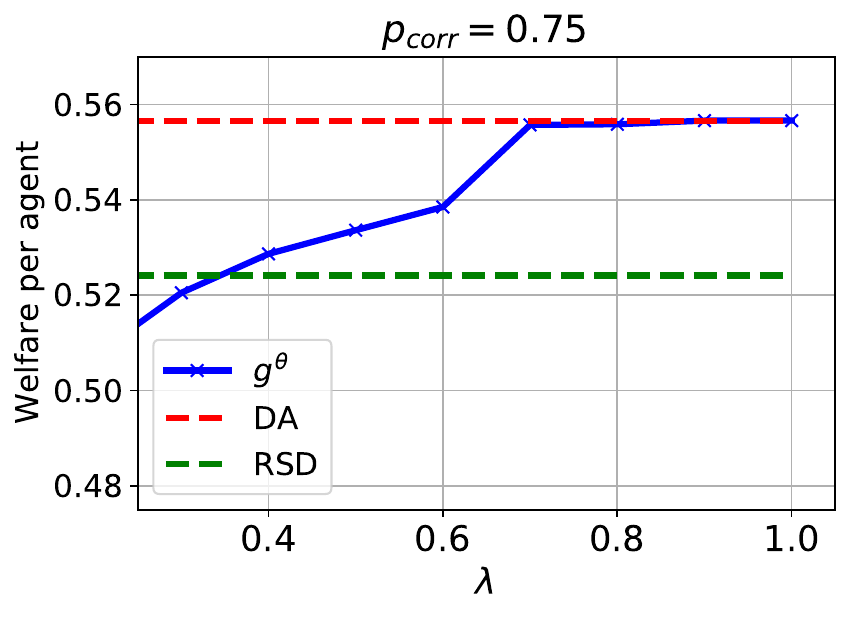}
\includegraphics[scale=0.33]{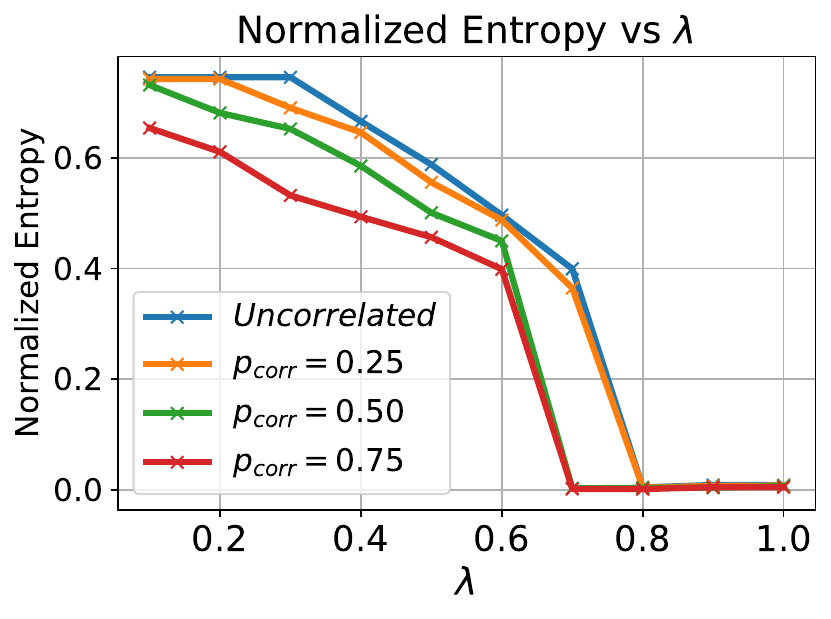}
\caption{\label{fig:props}
\textbf{Left and Middle}: Comparing welfare per agent of the learned mechanisms for different values of the tradeoff parameter $\lambda$ with the best of the firm- and worker- proposing DA, as well as
 TTC, and RSD. \textbf{Right, Top}: Comparing the average instance-wise max similarity scores ($\mathit{sim}(g^\theta)$) of the learned mechanisms with worker- and firm-proposing DA. \textbf{Right, Bottom}: Normalized entropy of the learned mechanisms for different values of the tradeoff parameter $\lambda$. The results are shown for uncorrelated preferences
 as well as an increasing correlation between preferences
 ($p_{\mathrm{corr}}\in \{0.25, 0.5, 0.75\}$)} 
\end{figure*}

Figure~\ref{fig:props} shows the expected welfare for the learned mechanisms, measured here for the equi-spaced utility function (the function used in the input representation). We define the welfare of a mechanism $g$ (for the equi-spaced utility function) on a profile $\succ$ as:
\begin{align}
 \mathit{welfare}(g,\succ)=\frac{1}{2}\left(\frac{1}{n} + \frac{1}{m}\right)\sum_{w\in W}\sum_{f\in F} g_{wf}(\succ) \left( p^{\succ}_{wf} + q^{\succ}_{wf} \right).
\end{align}
We compare against the maximum of the expected welfare achieved by the worker- and firm- proposing DA and TTC mechanisms, as well as that from RSD.
As we increase $\lambda$, and the learned mechanisms come closer to DA, the welfare of the learned mechanisms improves.
It is notable that for choices of $\lambda$ in the range 0.8 and higher, i.e., the choices of $\lambda$ that provide interesting opportunities for improving SP relative to DA, we also see good welfare. We also see that TTC, and especially RSD have comparably lower welfare.
 It bears emphasis that when a mechanism is not fully SP, as is the case for all mechanisms except RSD, 
this is an idealized view of welfare since it assumes truthful reports. In fact, we should expect welfare to be reduced through strategic misreports and the substantially improved SP properties of the learned mechanisms relative to DA (Figure~\ref{fig:frontier}) would be expected to further work in favor of improving the welfare in the learned mechanisms relative to DA.\footnote{In fact, the same is true for the stability of a non-SP mechanism such as DA, but it has become standard to assume truthful reports to DA in considering the stability properties of DA.}
Lastly, we observe that for small values of $\lambda$ the learned mechanisms have relatively low welfare compared to RSD. This is interesting and suggests that achieving IR together with SP (recall that RSD is not IR!) is very challenging in two-sided markets.
In interpreting the rules of the learned mechanisms, and considering the importance of DA, we can also compare their functional similarity with DA. For this, let $w$-DA and $f$-DA denote the worker- and firm-proposing DA, respectively. For a given preference profile, we compute the similarity of the learned
rule with DA as
\begin{equation}
 \mathit{sim}(g^\theta, \succ) = \max_{\mathcal{M} \in \{w\text{-DA}, f\text{-DA}\}}{\frac{\sum_{(w, f): g^\mathcal{M}_{wf}(\succ) = 1} g^\theta_{wf}(\succ)}{\sum_{(w, f): g^\mathcal{M}_{wf}(\succ) = 1} 1}}. 
\end{equation}

This calculates the agreement between the two mechanisms, normalized by the size of the DA matching,
 and taking the best of $w$-DA or $f$-DA. Let $\mathit{sim}(g^\theta)$
 denote the average similarity score on test data.
As we increase $\lambda$, i.e., penalize stability violations more, we see in 
Figure~\ref{fig:props} that the learned matchings get increasingly close to 
the DA matchings, as we might expect.
We also quantify the degree of randomness of the learned mechanisms,
 by computing the {\em normalized entropy per agent}, taking the expectation over all preference profiles. For a given profile $\succ$, we compute normalized entropy
per agent as (this is 0 for a deterministic mechanism): 
\begin{equation}
\begin{aligned}
 H(\succ) &= -\frac{1}{2n} \sum_{w \in W} \sum_{f \in \overline{F}} \frac{g_{wf}(\succ) \log_2 g_{wf}(\succ)}{\log_2 m} \\&- \frac{1}{2m} \sum_{f \in F} \sum_{w \in \overline{W}} \frac{g_{wf}(\succ) \log_2 g_{wf}(\succ)}{\log_2 n}.
\end{aligned}
\end{equation} 

Figure~\ref{fig:props} shows how the 
 entropy changes with $\lambda$.
 As we increase $\lambda$ and the mechanisms come closer to DA, the 
 allocations of the 
learned mechanisms
also becomes less stochastic.

\section{Discussion}
\label{sec:limit}

The methodology and results in this paper give a first 
but crucial step towards using machine learning to
understanding the structure of mechanisms that achieve nearly the same stability as DA while surpassing DA in terms of strategy-proofness. This is 
an interesting observation, given the practical 
and theoretical importance of the DA mechanism. 
Our experimental results
also suggest that achieving IR together with SP (recall that RSD is not IR) is challenging in two-sided markets,
 and this is reflected in the lower welfare
achieved by our learning mechanisms 
in the part of the frontier that emphasizes SP. The experimental results also show that for larger 
weights assigned to the importance of stability,
the learned mechanisms recover the 
DA mechanism, whereas when more emphasis 
is given to the importance of strategy-proofness,
the learned mechanisms are more randomized. There are other interesting questions waiting to be addressed. For instance, can we use this kind of framework to understand other tradeoffs, such as 
tradeoffs between strategy-proofness and efficiency?

As discussed in the introduction, a challenge in scaling to larger problems is the
need to find defeating misreports,
as exhaustively enumerating all misreports 
for an agent
becomes intractable as the number of agents on the other side of the market
increases (and thus the preference domain increases). 
A simple remedy is to restrict the language available to agents in making preference reports; e.g., 
it is commonplace to only allow for ``top-$k$ preferences" to be reported. Another remedy is to work in domains where there exists some structure on the preference domain, so that not all possible preference orders exist; e.g., {\em single-peaked preferences} are an especially stark example~\citep{black_rationale_1948}.
It will also be interesting to study complementary approaches that relax the discrete set of preference orderings to a continuous convex hull such as the {\em Birkhoff polytope}. Confining these preference orderings to the Birkhoff polytope can be accomplished using differentiable operations, such as the sinkhorn operator~\citep{adams2011ranking}. Such continuous relaxation allows misreports to be identified through gradient ascent rather than sampling (and would be similar to working with continuous valuations when learning revenue-optimal auctions~\citep{deep-auction19}). 
Despite this limitation, our current approach scales much further than other, existing methods for automated design, which are not well suited for this problem. For instance, methods that use linear programs or integer programs 
do not scale well because of the number of variables\footnote{$(m + 1)!^{n} \cdot (n + 1)!^{m} \cdot nm $ output variables} required to make explicit the input and output structure of the functional that must be optimized over.

A second
challenge is that 
we have not been able to find a suitable, publicly available dataset to test our approach. As a fallback, we have endeavored to capture some real-world structures by varying the correlation between
agent preferences
and the truncation probabilities of preferences. Using such stylized, probabilistic models and simulations for validating approaches is a well-established and prevalent practice, consistently utilized 
when investigating two-sided matching markets~\citep{CHEN2006202, RePEc:pri:econom:2013-3,das2005two,pmlr-v108-liu20c,dai2021learning}. For instance, Chen and S\"{o}nmez~\citep{CHEN2006202} design an environment for school choice where they consider six different schools with six seats each
and where the students' preferences 
are simulated to depend on proximity, quality, and a random factor. Echenique an Yariv~\citep{RePEc:pri:econom:2013-3} use a simulation study with eight participants on each side of the market, with the payoff matrix designed such that there are one, two, or three stable matches. Further, recent papers on bandit models for stable matching
model agent preferences through synthetic datasets~\citep{das2005two,pmlr-v108-liu20c,dai2021learning}.

In closing, we see exciting work ahead in advancing the design of matching mechanisms that strike the right balance between stability, strategyproofness, and other considerations that are critical to real-world applications.
As an example, it will be interesting to extend the learning framework to encompass desiderata such as capacity limitations or fairness considerations.
 Also relevant is to explore additional kinds of neural network architectures, such as
{\em attention mechanisms}~\citep{vaswani2017attention} or architectures that incorporate {\em permutation equivariance}~\citep{rahme2020permutationequivariant, pmlr-v162-duan22a}, with the advantage that they can
 significantly reduce the search space.

\section*{Acknowledgements}

The source code for all experiments along with the instructions to run it is available from Github at \url{https://github.com/saisrivatsan/deep-matching/}. This work is supported in part through an AWS Machine Learning Research Award. 

\bibliographystyle{abbrvnat}
\bibliography{references}

\appendix

\onecolumn

\section{RSD is ordinal SP but not Stable}\label{app:rsd-unstable}

We first show RSD satisfies FOSD and is thus ordinal SP. Consider agent $i$ in some position $k$ in the order. The agent's report has no effect on the choices of preceding agents, whether workers or firms (including whether agent $i$ is selected by an agent on the other side). Reporting its true preference ensures, in the event that it remains unmatched by position $k$, that it is matched with its most preferred agent of those remaining. 
For the one-sided version, the same argument holds for agents that are in the priority order. If an agent isn't on the side that's on the priority order, then that agent's report has no effect at all.

In the following example, we show RSD mechanism is not stable.

\begin{example}
\label{ex:1}
Consider $n = 3$ workers and $m = 3$ firms with the following preference orders:
\begin{align*}
 w_1: f_2, f_3, f_1, \bot \quad f_1: w_1, w_2, w_3, \bot\\
 w_2: f_2, f_1, f_3, \bot \quad f_2: w_2, w_3, w_1, \bot\\
 w_3: f_1, f_3, f_2, \bot \quad f_3: w_3, w_1, w_2, \bot
\end{align*}
The matching found by worker-proposing DA is $(w_1, f_3), (w_2, f_2), (w_3, f_1)$. This is a stable matching. If $f_1$ truncates and misreports its preference as $f_1: w_1, w_2, \bot, w_3$, the matching found is $(w_1, f_1), (w_2, f_2), (w_3, f_3)$. Firm $f_1$ is matched with a more preferred worker, and hence the mechanism is not strategy-proof.
Now consider the matching under RSD. The marginal matching probabilities $r$ is given by: 
\begin{equation*}
r = 
\left(\begin{matrix}
\frac{11}{24} & \frac{1}{4} & \frac{7}{24}\\[4pt]
\frac{1}{6} & \frac{3}{4} & \frac{1}{12}\\[4pt]
\frac{3}{8} & 0 & \frac{5}{8} 
\end{matrix}\right)
\end{equation*}

$f_2$ and $w_2$ are the most preferred options for $w_2$ and $f_2$ respectively and they would prefer to be matched with each other always rather than being fractionally matched with each other. Here $(w_2, f_2)$ is a blocking pair and thus RSD is not stable.
\end{example}

\section{TTC is neither Stable nor Strategyproof}\label{app:ttc-unstable}

\begin{example}
\label{ex:2}
Consider $n = 4$ workers and $m = 4$ firms with the following preference orders:

\begin{align*}
 w_1: f_1, \bot &\quad f_1: w_2, w_3, w_4, \bot\\
 w_2: f_2, \bot &\quad f_2: w_1 \bot\\
 w_3: f_1, \bot &\quad f_3: w_3 \bot\\
 w_4: f_3, \bot &\quad f_4: \bot\\
\end{align*}

\begin{figure}[t]
 \centering

 \begin{tikzpicture}[scale=0.7, transform shape, shorten >=1pt,->,draw=black!100, node distance=\layersep, thick]
 
 \tikzstyle{state}=[circle,draw=black!100,minimum size=25pt,inner sep=0pt,thick]

 \node[state] (F-1) at (0, 0) {$f_1$};
 \node[state] (F-2) at (-4, 0) {$f_2$};
 \node[state] (F-3) at (4, 0) {$f_3$};
 \node[state] (F-4) at (6, 0) {$f_4$};

 \node[state] (W-1) at (-2, 2) {$w_1$};
 \node[state] (W-2) at (-2, -2) {$w_2$};
 \node[state] (W-3) at (2, 2) {$w_3$};
 \node[state] (W-4) at (2, -2) {$w_4$};

 \path (W-1) edge (F-1);
 \path (W-2) edge (F-2);
 \path (W-3) edge (F-1);
 \path (W-4) edge (F-3);

 \path (F-1) edge (W-2);
 \path (F-2) edge (W-1);
 \path (F-3) edge (W-3);
 \path (F-4) edge [loop above] (F-4);

 \path (F-1) edge [dashed] (W-4);

 \end{tikzpicture}
 \caption{Round 1 of TTC. The solid lines represent workers and firms pointing to their top preferred agent truthfully. The dashed line represents a misreport by $f_1$}
 \end{figure}

 If all agents report truthfully, $w_1$ is matched with $f_1$. This violates IR as $\bot \succ_{f_1} w_1$ and thus the matching is not ex-ante stable. If $f_1$ misreports its preference as $f_1: w_4, w_3,\bot$, then $w_3$ is matched with $f_1$. Since $f_1$ is matched with a more preferred worker $w_3$ with $w_3 \succ_{f_1} \bot \succ_{f_1} w_1$, TTC is not strategyproof.

\end{example}
\section{Proof of Theorem ~\ref{thm:stability-violation}}\label{app:stv-proof}
\begin{proof}
Since $\mathit{stv}(g^\theta, \succ)\geq 0$ then $\mathit{STV}(g^\theta)=\mathbb{E}_{\succ} \mathit{stv}(g^\theta,\succ)=0$ if and only if $\mathit{stv}(g^\theta, \succ)= 0$ except on zero measure events. 
Moreover, $\mathit{stv}(g^\theta, \succ)= 0$ implies $\mathit{stv}_{wf}(g^\theta,\succ)=0$ for all $w\in W$, all $f\in F$. This is equivalent to no justified envy.
For firm $f$, this means $\forall w'\neq w$, $q^\succ_{wf} \leq q^\succ_{w'f}$ if $g^\theta_{w'f} > 0$ and $q\succ_{wf} \leq 0$ if $g^\theta_{\bot f} > 0$. Then there is no justified envy for a firm $f$. Analogously, there is no justified envy for worker $w$. If $g^\theta$ is {\em ex ante} stable, it trivially implies $STV(g^\theta)=0$ by definition.
\end{proof}

\section{Proof of Theorem~\ref{thm:regret0}}\label{app:sp-proof}
\begin{proof}
Since $\mathit{regret}(g^\theta, \succ)\geq 0$ then $\mathit{RGT}(g^\theta)=\mathbb{E}_{\succ} \mathit{regret}(g^\theta,\succ)=0$ if and only if 
$\mathit{regret}(g^\theta, \succ)= 0$ \sloppy except on zero measure events. 
Moreover, $\mathit{regret}(g^\theta, \succ)= 0$ implies $\mathit{regret}_w(g^\theta, \succ) = 0$ for any worker $w$ and $\mathit{regret}_f(g^\theta, \succ) = 0$ for any firm $f$. Thus, the maximum utility increase on misreporting is at most zero, and hence $g^\theta$ is ordinal SP. If $g^\theta$ is Ordinal-SP, it is also satisfies FOSD~\ref{thm:fosd} and it is straightforward to show that $\mathit{regret}(g^\theta, \succ)= 0$.
\end{proof}

\section{Training Details and Hyperparameters}\label{app:train-detail}

We use a neural network with $R = 4$ hidden layers with $256$ hidden units each for all our settings. We use the leaky ReLU activation function at each of these layers. To train our neural network, we use the Adam Optimizer with decoupled weight delay regularization (implemented as {\em AdamW} optimizer in PyTorch) We set the learning rate to $0.005$ for uncorrelated preferences setting and $0.002$ when $p_{corr} = \{0.25, 0.5, 0.75\}$. The remaining hyperparameters of the optimizer are set to their default values. We sample a fresh minibatch of 1024 profiles and train our neural networks for a total of 50000 minibatch iterations. We reduce the learning rate by half once at $10000^{th}$ iteration and once at $25000^{th}$ iteration. We report our results on 204800 preference profiles. For our training, we use a single Tesla V-100 GPU. For each setting, the neural network takes 4.6 hours to train.

\end{document}